\documentclass[11pt]{article}
\usepackage[utf8]{inputenc}
\usepackage[top=1in, bottom=1in, left=1.25in, right=1.25in]{geometry}
\usepackage{amsmath}
\usepackage{amssymb}
\usepackage{amsthm}
\usepackage{thmtools}
\usepackage{tikz}
\usetikzlibrary{arrows.meta}

\usepackage{enumitem}
\usepackage{xcolor}
\usepackage{hyperref}
\usepackage[langfont=typewriter]{complexity}
\usepackage[capitalize]{cleveref} 
\usepackage{microtype}
\usepackage{authblk}

\declaretheorem[qed={$\lrcorner$}, numberwithin=section]{theorem}
\declaretheorem[qed={$\lrcorner$}, sibling=theorem]{proposition}

\theoremstyle{definition}
\declaretheorem[qed={$\lrcorner$}, sibling=theorem]{definition}

\DeclareMathOperator{\dom}{dom}

\usepackage{caption}

\DeclareCaptionFormat{myformat}{#1#2#3\hrulefill}
\title{E-Graphs as Circuits, and Optimal Extraction via Treewidth\footnote{This material is based upon work supported by the National Science Foundation under Grant Nos. CCF-1836724, CCF-2006359, CNS-2232339, and CCF-2312195. Any opinions, findings, and conclusions or recommendations expressed in this material are those of the authors and do not necessarily reflect the views of the National Science Foundation.}}

\author{Glenn Sun}
\author{Yihong Zhang}
\author{Haobin Ni}
\affil{University of Washington}
\date{}

\begin{document}

\maketitle

\begin{abstract}
    We demonstrate a new connection between e-graphs and Boolean circuits. This allows us to adapt existing literature on circuits to easily arrive at an algorithm for optimal e-graph extraction, parameterized by treewidth, which runs in $2^{O(w^2)}\text{poly}(w, n)$ time, where $w$ is the treewidth of the e-graph. Additionally, we show how the circuit view of e-graphs allows us to apply powerful simplification techniques, and we analyze a dataset of e-graphs to show that these techniques can reduce e-graph size and treewidth by 40-80\% in many cases. While the core parameterized algorithm may be adapted to work directly on e-graphs, the primary value of the circuit view is in allowing the transfer of ideas from the well-established field of circuits to e-graphs.
\end{abstract}

\section{Introduction}\label{introduction}

E-graphs are a type of directed graph with an equivalence relation on its nodes that can be used to compactly represent exponentially many equivalent expressions. In recent years, this capability has found e-graphs to have many applications in formal methods, compilers, and automated reasoning communities (\cite{jnr02}, \cite{tstl09}, \cite{stl11}, \cite{w+21}).

One important problem with e-graphs is extraction: from the compact representation, how does one pick a minimum cost expression? E-graph extraction is known to be NP-hard \cite{ste11}, so applications of e-graphs often use suboptimal techniques like greedy algorithms \cite{w+21} to extract one expression out of an e-graph. If one is interested in exact optimal extraction, integer linear programming (ILP) is sometimes used \cite{y+21}, but there has been little research into specialized algorithms to solve extraction optimally.

Common algorithmic techniques for NP-hard optimization problems include approximation algorithms and parameterized algorithms. It turns out that extraction is also hard to approximate to any constant factor \cite{glp24}, so it is natural to turn to parameterized algorithms: algorithms that are efficient after a particular parameter is chosen to be fixed. One commonly used parameter is \emph{treewidth}, a measure of how ``close'' to a tree the graph is (a survey is available in \cite{bod06}).

We make three key observations:
\begin{enumerate}
    \item E-graphs that appear in practice often have low treewidth, so it is a good choice for parameterization. We quantify this in \cref{simplification-evaluation}. 
    \item E-graphs may be considered a certain class of \emph{monotone circuits}, meaning a Boolean circuit with only AND and OR gates (no NOT gates). This allows us to draw on existing literature about treewidth-based algorithms for circuits in order to solve extraction. 
    \item The circuit view also allows us to apply circuit simplification, which is more broadly studied and flexible than simplifying e-graphs directly.
\end{enumerate}

More specifically, we will show in \cref{circuits-and-e-graphs} that with the appropriate translation between e-graphs and circuits, the extraction problem is nothing more than the weighted monotone circuit satisfiability problem, with the one caveat that our circuits may have cycles. In \cref{the-main-dynamic-programming-algorithm}, we then draw on an existing algorithm for weighted \emph{acyclic} monotone circuit satisfiability given by Kanj, Thilikos, and Xia \cite{ktx17}. Their algorithm parameterizes on treewidth as desired, and we make only one minor change to take care of cycles. In \cref{circuit-simplification}, we discuss simplification rules that make the main algorithm more practical.

Simultaneously and independently of our own efforts, Goharshady, Lam, and Parreaux \cite{glp24} also gave a solution to the e-graph extraction problem, also through parameterization by treewidth. They do not use circuits, but our main algorithms are actually very similar. They include a few additional extensions and optimizations as well. 

Our main contribution is the connection between circuits and e-graphs, which not only allows us to use existing circuit algorithms but also employ simplification techniques more effectively. Given that our main algorithm is extremely similar to the one given in \cite{glp24}, we will omit to give our own practical implementation, evaluation against existing extraction methods, and proofs of correctness for the main algorithm, and direct the interested reader to their paper. While we do have a publicly available implementation\footnote{\url{https://github.com/glenn-sun/egg-extraction-gym/tree/glenn-treewidth/src/extract/treewidth}}, our focus is on the circuit connection and simplification.

\begin{figure}
    \centering
    \begin{tikzpicture}
        \draw (-2, 5) -- (-1.5, 6.5) -- (2, 6.5) -- (1.5, 5) -- (-2, 5);
        \node[align=center] at (0, 5.75) {Input e-graph};

        \draw[->, -{Stealth[scale=1.5]}] (0, 5) -- (0, 4);
        
        \draw (-2, 2.5) rectangle (2, 4);
        \node[align=center] at (0, 3.25) {Convert to circuit\\\emph{\cref{circuits-and-e-graphs}}};

        \draw[->, -{Stealth[scale=1.5]}] (0, 2.5) -- (0, 1.5);
        
        \draw (-2, 0) rectangle (2, 1.5);
        \node[align=center] at (0, 0.75) {Circuit simplification\\\emph{\cref{circuit-simplification}}};

        \draw[->, -{Stealth[scale=1.5]}] (0, 0) -- (0, -0.5) -- (4, -0.5) -- (4, 7) -- (8, 7) -- (8, 6.5);

        \draw (6, 5) rectangle (10, 6.5);
        \node[align=center] at (8, 5.75) {Tree decomposition\\\emph{external libraries}};

        \draw[->, -{Stealth[scale=1.5]}] (8, 5) -- (8, 4);
        
        \draw (6, 2.5) rectangle (10, 4);
        \node[align=center] at (8, 3.25) {Main algorithm\\\emph{\cref{the-main-dynamic-programming-algorithm}}};

        \draw[->, -{Stealth[scale=1.5]}] (8, 2.5) -- (8, 1.5);

        \draw (6, 0) -- (6.5, 1.5) -- (10, 1.5) -- (9.5, 0) -- (6, 0);
        \node[align=center] at (8, 0.75) {Output extraction};
    \end{tikzpicture}
    \caption{Overall algorithm pipeline and article organization}
    \label{fig:pipeline}
\end{figure}
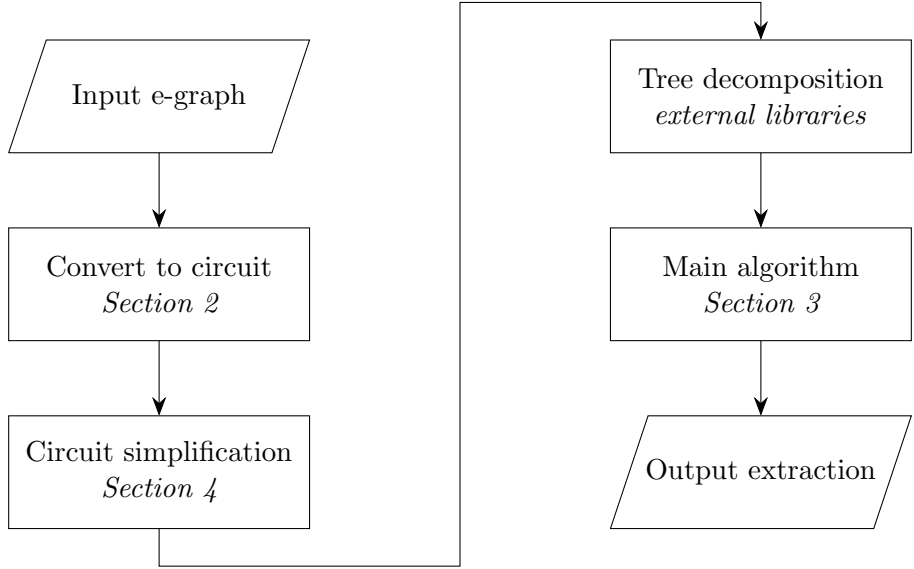















\section{Circuits and e-graphs}\label{circuits-and-e-graphs}

We first show the equivalence of circuits and e-graphs. For the reader who is already familiar with e-graphs, you may find \cref{fig:equiv} sufficient to give the intuition of the equivalence.

\begin{definition}
    A \emph{weighted cyclic monotone circuit} (henceforth ``circuit'') is a directed graph with additional information $G = (V, E, V_{out}, g, c)$. The set $V_{in} \subseteq V$ is the set of \emph{inputs}, defined to be the set of vertices with in-degree 0. The set $V_{out} \subseteq V$ is the set of \emph{outputs}. (These are the nodes whose values we are interested in, which may or may not have out-degree 0.) Finally, $g : V \setminus V_{in} \to \{\text{AND}, \text{OR}\}$ is the gate type function and $c : V_{in} \to \mathbb{R}$ is a cost function. For any $U \subseteq V_{in}$, we denote $c(U) = \sum_{u \in U} c(u)$.
\end{definition}

Because our circuits have cycles, we need to be a bit more precise than usual about the semantics of the circuit. In particular, there may be undefined behavior on certain inputs.

\begin{definition}
    A function $\alpha : V \to \{0, 1\}$ is a (valid total) \emph{evaluation} of $G$ if for all $u \in V$, denoting the inputs to $u$ as $\{v_1, \dots, v_k\}$, we have $\alpha(u) = g(u)(\alpha(v_1), \dots, \alpha(v_k))$. The evaluation \emph{satisfies $G$} if $\alpha(u) = 1$ for all $u \in V_{out}$. It \emph{minimally satisfies $G$} if for every proper subset $A \subsetneq \{u \in V : \alpha(u) = 1\}$, the function \[
        \alpha|_{A}(u) = \begin{cases} \alpha(u) & \text{if $u \in A$} \\  0 & \text{if $u \notin A$} \end{cases}
    \] 
    is not a valid evaluation satisfying $G$. We denote $G[\alpha]$ to be the subgraph of $G$ induced by the set $\{u \in V : \alpha(u) = 1\}$. We say that $\alpha$ is \emph{acyclic} if $G[\alpha]$ is acyclic.
\end{definition}

The key fact that allows us to have well-defined semantics with cyclic circuits is the following:

\begin{proposition}
    An acyclic evaluation $\alpha$ is uniquely determined by its value on the inputs.
\end{proposition}
\begin{proof}
     Let $\beta$ be an acyclic evaluation which agrees with $\alpha$ on $V_{in}$, we will show that $\alpha(u) = 1$ iff $\beta(u) = 1$. The forward and backward directions are identical. Let us treat the forward direction.

    Recall that $\alpha(u) = 1$ iff $u \in G[\alpha]$. Because $G[\alpha]$ is acyclic, take the vertices in topological order. The base cases are the elements of $G[\alpha]$ with in-degree 0; note that these must have in-degree 0 in $G$ because gates cannot be 1 without at least one 1 input. Hence $\beta$ agrees with $\alpha$ here by hypothesis.
    
    For the inductive step, to show that $\beta(u) = 1$, take cases based on the gate type of $u$. If $u$ is an AND gate, because $\alpha(u) = 1$, all of $u$'s inputs in $G$ belonged to $G[\alpha]$, so $\beta$ is 1 there by induction, and the only valid choice for $\beta(u)$ is 1. If $u$ is an OR gate, a similar argument applies. 
\end{proof}

Next, let us draw the connection between e-graphs and circuits.

\begin{definition}
    An \emph{e-graph} is a structure $\mathcal{G} = (N, \mathcal{C}, \mathcal{E}, \mathcal{C}_{out}, c)$, where $N$ is a set of \emph{e-nodes}, $\mathcal{C}$ is a partition of $N$ into \emph{e-classes}, $\mathcal{E} \subseteq  N \times \mathcal{C}$ is a directed edge relation, $\mathcal{C}_{out} \subseteq \mathcal{C}$ is the set of \emph{output classes}, and $c : N \to \mathbb{R}$ is a cost function. For any $M \subseteq N$, we denote $c(M) = \sum_{u \in M} c(u)$.
\end{definition}
\begin{definition}
    Let $\dom(\varphi) \subseteq \mathcal{C}$. An \emph{extraction} of an e-graph is a choice function $\varphi : \dom(\varphi) \to N$ (that is, $\varphi(C) \in C$ for all $C \in \dom(\varphi)$) which additionally has that whenever $C \in \dom(\varphi)$, the inputs to $\varphi(C)$ are all in $\dom(\varphi)$ as well. The extraction is \emph{satisfying} if $C \in \dom(\varphi)$ for all $C \in \mathcal{C}_{out}$. It is \emph{minimally satisfying} if for every proper subset $\mathcal{A} \subsetneq \dom(\varphi)$, the function $\varphi|_{\mathcal{A}}$ is not a satisfying extraction. A \emph{selected path} in $\varphi$ is a finite list of e-classes $C_1, \dots, C_k$ such that for all $1 \le i \le k-1$, we have $(\varphi(C_{i}), C_{i+1}) \in \mathcal{E}$. The extraction is \emph{acyclic} if there are no selected paths from a class to itself.
\end{definition}

Our main observation is that every e-graph can be represented as a circuit in such a way that its semantics are equivalent. For an example, see \cref{fig:equiv}.

\begin{figure}
    \centering
    \includegraphics[scale=0.8667]{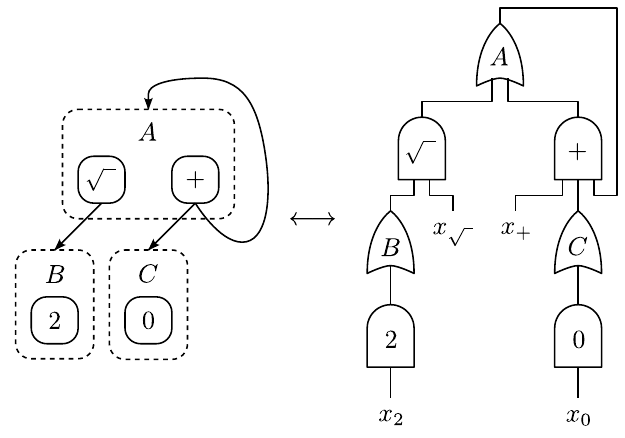}
    \caption{An example of converting e-graphs into circuits. Note that the arrows conventionally point to dependencies in e-graphs, but signals flow in the opposite direction in a circuit, so we flip the arrows. Furthermore, the extraction $A \mapsto \sqrt{\phantom{x}}$, $B \mapsto 2$ corresponds to the evaluation where everything on the left and the OR gate $A$ are all 1, and the rest are 0. The cyclic extraction $A \mapsto +$, $C \mapsto 0$ corresponds to the cyclic evaluation where everything on the right and the OR gate $A$ are all 1, and the rest are 0. The evaluation where everything is 1 has no corresponding extraction because $A$ can only choose one e-node in an e-graph; however, such an evaluation is not minimal.}
    \label{fig:equiv}
\end{figure}


\begin{proposition} \label{prop:equiv}
     Given an e-graph $\mathcal{G} = (N, \mathcal{C}, \mathcal{E}, \mathcal{C}_{out}, c)$, construct a monotone circuit $G = (V, E, V_{out}, g, c)$ by converting every e-class into an OR gate, every e-node into an AND gate, and then flip every edge. Additionally, create one input for every e-node, and attach it to its corresponding AND gate, with cost set equal to the cost of the e-node.
     
     Then, there exists an acyclicity-preserving bijection between minimal satisfying extractions of $\mathcal{G}$ and minimal satisfying evaluations of $G$.
\end{proposition}

The proof is a tedious checking of these definitions that does not require any external results. The details are contained in \cref{sec:appendix}. With this observation, in order to solve the extraction problem for e-graphs, it suffices to solve the weighted satisfiability problem for (potentially cyclic) circuits.

\section{The main dynamic programming algorithm}\label{the-main-dynamic-programming-algorithm}

\subsection{Preliminaries}\label{preliminaries}

Our main inspiration is Proposition 4.7 of \cite{ktx17}, which solved the weighted minimum satisfiability problem for acyclic monotone circuits by parameterizing on treewidth. Based on the reduction illustrated in the previous section, the only additional algorithmic contribution that we need to make is to describe why the cyclic nature of the graph is not a problem and how we enforce the extraction result to be acyclic. For completeness, we will recap the full algorithm (with slightly different notations from the original paper).

The main technique is called treewidth, or tree decomposition. This is a classical technique; see Chapter 7 of \cite{c+15} for more information. 

\begin{definition}
    Given a undirected graph $G = (V, E)$, a tree decomposition of $G$ is a tree $\mathcal{T} = (\mathcal{X}, \mathcal{E})$, whose vertices are subsets of $V$ (called \emph{bags}), satisfying:

\begin{enumerate}
\def\labelenumi{\arabic{enumi}.}
\item
  For all $\{u, v\} \in E$, there exists a bag $X \in \mathcal{X}$ such that $u, v \in X$.
\item
  For all $v \in V$, the subgraph of $\mathcal{T}$ induced by the bags that contain $v$ forms a tree.
\end{enumerate}

The \emph{width} of a tree decomposition is the size of the largest bag minus one. The \emph{treewidth} of a graph is the smallest width of any tree decomposition. If $\mathcal{T}$ is rooted, we write $T_X$ for the union of all bags underneath $X \in \mathcal{X}$ (including $X$). A \emph{nice tree decomposition} is one in which every bag $X$ is one of 4 kinds:

\begin{enumerate}
\def\labelenumi{\arabic{enumi}.}

\item
  Leaf bag: $X = \emptyset$.
\item
  Insert bag: $X = Y \cup \{u\}$, where $Y$ is the unique child of $X$ and $u \in V \setminus Y$.
\item
  Forget bag: $X = Y \setminus \{u\}$, where $Y$ is the unique child of $X$ and $u \in Y$.
\item
  Join bag: $X$ has two children, which contain exactly the same vertices as $X$. \qedhere
\end{enumerate}
\end{definition}

\begin{figure}
    \centering
    \includegraphics[scale=0.8667]{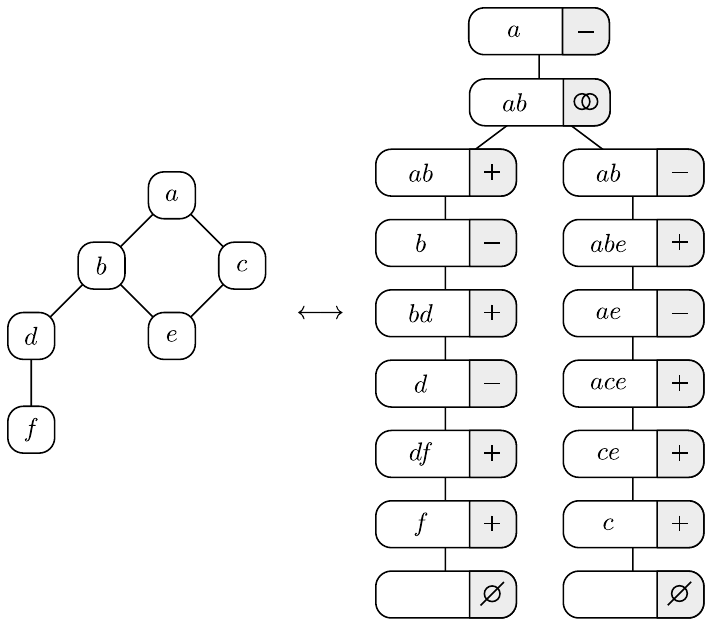}
    \caption{A nice tree decomposition of a graph, where each node is annotated whether it is a leaf, insert, forget, or join node. This graph has treewidth 2.}
    \label{fig:td}
\end{figure}

A nice tree decomposition can be computed from a tree decomposition in linear time: pick any bag to serve as the root, then for all children of a bag, first forget and insert the difference between the child and the current bag, then join all of the copies of the current bag, and repeat. 

The core feature of tree decomposition is that as you walk up the tree, the current bag cuts the original graph into two disconnected pieces: the set of vertices that you have already seen (including those that are forgotten) and the set of vertices that you have not seen yet. This property allows us to do dynamic programming and focus only on the current bag.

One last definition that we will use in the algorithm is that of a partial evaluation since we would like to build up our evaluations at each bag incrementally. We will use the following definition with $U = T_X$.

\begin{definition}
    A (locally valid) partial evaluation of a circuit is a function $\alpha : U \to \{0, 1\}$ where $U \subseteq V$, and for all $u \in U$, denoting the inputs to $u$ as $\{v_1, \dots, v_k, w_1, \dots, w_{\ell}\}$ where each $v_i \in U$ and each $w_j \not \in U$, there exist $b_1, \dots, b_\ell \in \{0, 1\}$ such that $\alpha(u)$ is $g(u)$ (AND or OR) applied to $(\alpha(v_1), \dots, \alpha(v_k), b_1, \dots, b_{\ell})$.
\end{definition}

\subsection{The algorithm}\label{the-algorithm}

We are given a weighted cyclic monotone circuit $G = (V, E, V_{out}, g, c)$, and are tasked to compute a minimum cost satisfying evaluation. The algorithm sketch is as follows:

\begin{enumerate}
\def\labelenumi{\arabic{enumi}.}
\item
  Add a new AND gate to $G$, with every vertex in $V_{out}$ as inputs, and call it $u_{out}$.
\item
  Compute a tree decomposition of the undirected underlying graph of $G$. 
\item
  Compute a nice tree decomposition rooted at any bag containing $u_{out}$.
\item
  We do dynamic programming. At every bag, for every possible \emph{summary} of a partial evaluation, the program will remember the minimum cost partial evaluation producing that summary. The summary of a partial evaluation $\alpha : T_X \to \{0, 1\}$ at bag $X$ has three parts:
  \begin{itemize}
      \item First, there is the restriction $\alpha|_X$.
      \item Second, there is the map $\text{known}_\alpha|_X$, where $\text{known}_\alpha(u) = 1$ if and only if $\alpha(u) = 1$ and that this fact can be deduced from the value of $\alpha$ on the inputs to $u$. For example, if $u$ is an AND gate and one of its inputs is not in $\dom(\alpha)$, then $\text{known}_\alpha(u) = 0$.
      \item Lastly, there is $G[\alpha]^+[X]$, the subgraph of the transitive closure of $G[\alpha]$ induced by $X$. In other words, which elements in $X$ have paths between them in $G[\alpha]$? This is what allows us to correctly handle acyclicity, remembering which vertices in the bag are connected, even when the connections themselves have already been forgotten from the bag, and is the only change from \cite{ktx17}.
  \end{itemize}  
    We break into cases depending on the type of bag.
  \begin{enumerate}
  \def\labelenumii{\arabic{enumii}.}
  
  \item
    Leaf bag: We map the empty summary to the empty evaluation.
  \item
    Insert bag ($X = Y \cup \{u\}$): For every partial evaluation $\alpha$ remembered at $Y$, we attempt to extend it with $\alpha(u) = 0$ and $\alpha(u) = 1$. If these extensions are valid and acyclic, compute their new summaries and remember them if they have the smallest cost for their summary so far.
  \item
    Forget bag ($X = Y \setminus \{u\}$): For every partial evaluation $\alpha$ remembered at $Y$, restrict the entire summary to $X$ and remember the lowest cost evaluations per summary.
  \item
    Join bag: Denote this bag as $X$ with children $Y$ and $Z$ (even though as sets $X = Y = Z$). By property 2 of tree decompositions, $T_Y$ and $T_Z$ intersect only at $X$. Therefore, for all remembered $\alpha : T_Y \to \{0, 1\}$ and $\beta : T_Z \to \{0, 1\}$, as long as they agree on $X$, they can be merged into a new evaluation on $T_X$. If it is valid and acyclic, compute its summary and remember it if it has the smallest cost for its summary so far.
  \end{enumerate}
\item
  At the root, output the minimum cost evaluation producing the summary corresponding to $\alpha(u_{out}) = 1$.
\end{enumerate}

The discussion within the algorithm gave some ideas to why $G[\alpha]^+[X]$ is necessary in the summary, but it remains to motivate $\text{known}_\alpha$. Consider an OR gate $u$ such that all but one of its children have been forgotten. Without $\text{known}_\alpha$, we might keep only the evaluation where all of its children are set to 0, because that is cheaper, forcing us to pick the last child to continue to this line, even if it is suboptimal. 

The running time of this algorithm is roughly $2^{O(w^2)}\text{poly}(w, n)$, where $w$ is the treewidth and $n = |V|$. The precise coefficients and polynomial degree are dependent on the data structures in the implementation, but the largest term comes from there being at most $2^{O(w^2)}$ distinct summaries for each bag, since the transitive graph has $O(w^2)$ edges. For more details, see \cite{ktx17} or \cite{glp24}.

Note that when implementing this algorithm, instead of actually remembering the best partial evaluation per summary, it is faster to just remember a pointer to the previous summary that produced it. At the end of the algorithm, one can walk back and recover the full evaluation using these pointers. Depending on the precise data structures, another small optimization could be to combine $\alpha|_X$ and $\text{known}_\alpha|_X$ into a ternary-valued collection in the summary, since $\text{known}_\alpha$ is only recorded for true vertices. 

\section{Circuit simplification}\label{circuit-simplification}

\subsection{Rules}\label{rules}

One natural way to improve the speed of our algorithm is to simplify the instances directly. We found in our testing that this helps dramatically. The idea of preprocessing e-graphs to simplify them before extraction is not new (see for example, the repository at \cite{han24}), but the circuit view of e-graphs makes this process easier, more transparent, and more powerful.

Finding the most compact representation of a Boolean circuit is often known as circuit minimization. Circuit minimization is well-known to be NP-hard, but there is an abundance of existing software to quickly attempt a best effort, such as SIS \cite{sis} among others. 

Note that in order to use off-the-shelf software, it must support cycles. Though our particular semantics are unique, all simplifications valid for \emph{sequential circuits} are valid for us, a type of circuits fundamental in hardware. Sequential circuits are cyclic circuits where the ``undefined'' behavior is explicitly specified by propagation delay and keeping track of the circuit's state. A minimizer for sequential circuits will additionally ensure that these stateful behaviors are preserved, which we can just ignore---it must also preserve behaviors independent of state, as acyclic extractions are. Thus, we may want to do some additional simplification afterwards specific to our own extraction semantics, but such software may be used as a first step.

To show some basic ideas, as well as to highlight properties of extraction to make simplifications beyond what is possible from off-the-shelf software, we implemented several of our own heuristic rules. For all of the rules below, let $G = (V, E, V_{out}, g, c)$ be a monotone circuit. These rules are by no means an exhaustive list of all possible simplifications, they are only an exploratory list of the kinds of rules that may be beneficial.

\begin{proposition}
    When applying each of the following rules to a weighted cyclic monotone circuit, the optimal acyclic evaluation is either retained or efficiently recoverable.
\textnormal{
    \begin{enumerate}
\def\labelenumi{\arabic{enumi}.}
\item
  (Remove unreachable) For all $u \in V$, if there does not exist a path from $u$ to some $v \in V_{out}$, then it is safe to remove $u$ from $V$.
\item
  (Contract indegree one) Suppose $u \in V$ has indegree 1, in particular $(v, u) \in E$. Then it is safe to contract the edge $(v, u)$. The new vertex has the same gate type as $v$.
\item
  (Contract same gate) Suppose $v \in V$ has outdegree 1, in particular $(v, u) \in E$, and suppose $g(v) = g(u)$. Then it is safe to contract the edge $(v, u)$. The new vertex has the same gate type as $v$ and $u$.
\item
  (Same gate no shortcut) Suppose $(v, u) \in E$ and there exists a path $(v, w_1), (w_1, w_2), \dots, \allowbreak (w_{n-1}, w_n), (w_n, u) \in E$ such that $g(v) = g(w_1) = \dots = g(w_{n-1}) = g(u)$. Then it is safe to delete $(v, u)$.
\item
  (Factoring) Suppose we have $(w, v_i), (v_i, u) \in E$ where $g(v_i) = \text{OR}$ for $2 \le i \le n$, and $g(u) = \text{AND}$. Then it is safe to delete all of these edges and replace them with two new vertices $a$ and $b$, where $g(a) = \text{AND}$ and $g(b) = \text{OR}$, with the edges $(v_i, a)$ for all $i$, $(a, b)$, $(w, b)$, and $(b, u)$. The rule may also apply with AND and OR swapped.
\item
  (Remove lone OR loops) Suppose $u \in V$ is an OR gate and $v_1, \dots, v_n \in V$ are AND gates, and $(u, v_1), (v_1, v_2), \dots, (v_{n-1}, v_n), (v_n, u) \in E$. Then it is safe to delete $v_n$.
\item
  (Collect variables) Suppose $u_1, u_2 \in V$ are variables with the same out-neighborhood, all of which are AND gates. Then it is safe to merge $u_1$ and $u_2$ into a new variable with the same out-neighborhood, with cost the sum of the originals.  \qedhere
\end{enumerate}
}
\end{proposition}
\begin{proof}
    \begin{enumerate}
        \item The optimal acyclic evaluation is minimal, so it does not set true any vertices that do not affect the output. 
        \item Suppose $u \in V$ has indegree 1, in particular $(v, u) \in E$. Then it is safe to contract the edge $(v, u)$. The new vertex has the same gate type as $v$.
        \item Because AND and OR are associative, when two of the same gate are adjacent and the subexpression is not reused in other situations, it is an equivalent circuit to merge the two gates. 
        \item One may check for both AND and OR that after deleting the edge, the dependency is still maintained through the path. 
        \item This rule is nothing more than observing how $(w \lor x_1) \land \dots \land (w \lor x_n) = w \lor (x_1 \land \dots \land x_n)$, generalized slightly.

        \begin{figure}
            \centering
            \includegraphics[scale=0.8667]{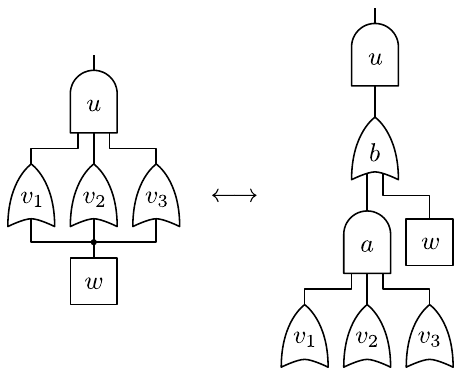}
            \caption{The factoring rule. All gates may have additional inputs/outputs, which are preserved and not depicted here. The square denotes any vertex. }
            \label{fig:factor}
        \end{figure}

        (See \cref{fig:factor} for a visualization of this rule. We note that this is the only rule that may increase the size of the circuit. However, it is generally beneficial to apply this because it reduces the number of cycles in the underlying undirected graph, which generally reduces the treewidth.)
        \item Suppose $v_n$ is true. Then by induction, every $v_i$ is true, as well as $u$. Then, this evaluation has a true cycle. So, the optimal acyclic evaluation must have $v_n$ false, and we can delete it. 
        \item Simply note that an evaluation that sets $u_1$ to true but $u_2$ to false is not minimal, so $u_1$ and $u_2$ must be both false or both true. \qedhere
    \end{enumerate}
\end{proof}


Some of these rules are purely based on the circuit structure (rules 2, 3, 4, 5). These were rules that we arbitrarily decided to implement, but in future work, one might consider replacing with an existing general purpose circuit minimization tool. This is one potential benefit of circuits that we have yet to explore. Other rules are more specific to extraction (rules 1, 6, 7). We note that some rules can be generalized to all extraction algorithms, such as rule 1.

Although there are often equivalent rules that operate on e-graphs directly without translating to circuits, circuits can generally be smaller (not to mention unlocking the ability to harness existing software). Recall that a direct translations of an e-graph will always result in an alternating AND/OR pattern, with every AND gate having outdegree 1 and one variable as an in-neighbor. Rules like rule 2 and rule 7 break this structure, and they are valid because the algorithm works for all monotone circuits.

\subsection{Simplification evaluation}\label{simplification-evaluation}

To evaluate our simplification rules, we applied them to a large set of e-graphs from various sources, collected in the ``extraction-gym'' benchmarking suite \cite{extraction-gym}. These are e-graphs that were generated by real projects, such as a e-graph based general purpose compiler (``eggcc-bril''), a compiler for specialized hardware (``flexc''), a fuzzer for automated testing (``fuzz''), and several other sources. A basic summary of the dataset is given in \cref{tbl:gym}. 

\begin{table}
\centering
\begin{tabular}{c||c|c|c}
\textbf{Source} & \textbf{no. of e-graphs} & \textbf{avg. $|V|$} & \textbf{avg. degree} \\ \hline \hline
babble          & 173                      & 5336.8              & 2.6                  \\ \hline
egg             & 28                       & 4276.5              & 4.1                  \\ \hline
eggcc-bril      & 36                       & 20329.5             & 2.7                  \\ \hline
flexc           & 14                       & 23620.7             & 3.4                  \\ \hline
fuzz            & 18                       & 126.0               & 4.0                  \\ \hline
rover           & 9                        & 21303.4             & 6.8                  \\ \hline
tensat          & 10                       & 57969.9             & 3.3                 
\end{tabular}
\caption{Basic characteristics of ``extraction gym'' dataset after circuit conversion.}
\label{tbl:gym}
\end{table}

Note that the main algorithm's correctness only relies on having a valid tree decomposition of any width---a larger width only affects the running time. Thus, although we unfortunately found that the vast majority of e-graphs in this test set were too large for computing exact tree decomposition, approximate tree decomposition algorithms suffice. We used the ``arboretum'' Rust library, which implements various heuristics but primarily relies on the classical minimum degree heuristic to compute an upper bound on treewidth. 

In order to compare the effects of our simplification scheme, such an upper bound is also a more realistic measure to compare than the true treewidth, since the important quantity is the width of the tree decomposition available to our algorithm. In our implementation, we apply each of the rules in a loop until we reach a fixed point. The effect of simplification on treewidth and $|V|$ is shown in \cref{fig:simp}, and these effects are quantified in \cref{tbl:simp}.

Note that every source produces e-graphs in a different way, which can dramatically affect which optimizations are more effective. E-graphs from some sources, like ``eggcc-bril'', demonstrated extraordinary simplification with $65\%$ reduction in treewidth and $97\%$ reduction in $|V|$, whereas e-graphs from other sources like ``babble'' or ``flexc'' demonstrated negligible improvement in treewidth (or even slight degradation due to the non-exact tree decomposition algorithm), although $|V|$ continues to be reduced substantially, by $60\%$ or more in all but one collection. 

We note that although the treewidth of most e-graphs in our dataset is small, sparsity is likely to be the only contributing factor to low treewidth in these data, not any deeper features of the e-graph generation process. This is because the linear relationship between $|V|$ and treewidth is the exact relationship predicted by random graphs of constant average degree: In a random graph where each edge has probability $c/n$ of appearing for $c > 1$, the treewidth of the graph is $\Omega(n)$ \cite{llo11} and upper bounded by $tn$ for some $t < 1$ \cite{wlcx11}.

\begin{figure}
    \centering
    \includegraphics[scale=0.8]{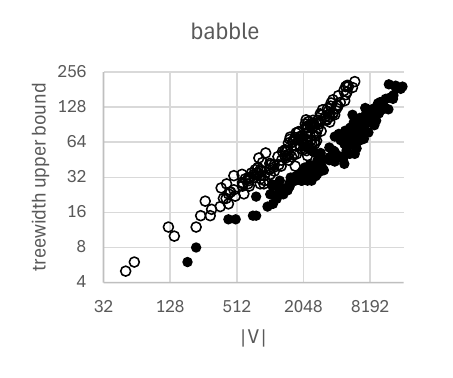}
    \includegraphics[scale=0.8]{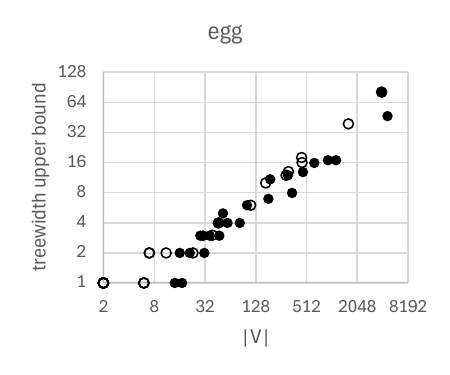}
    \includegraphics[scale=0.8]{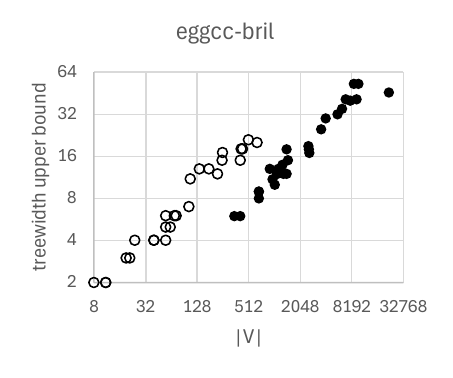}
    \includegraphics[scale=0.8]{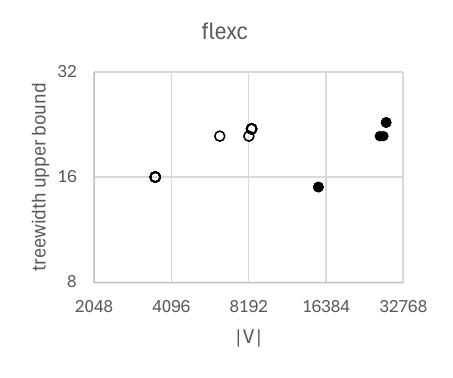}
    \includegraphics[scale=0.8]{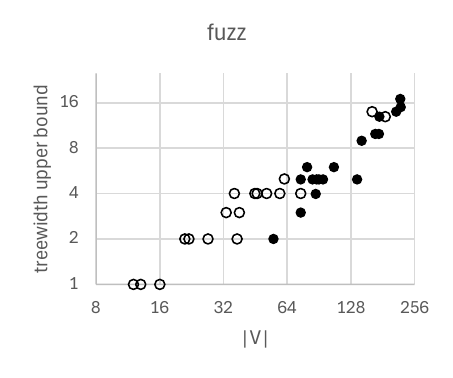}
    \includegraphics[scale=0.8]{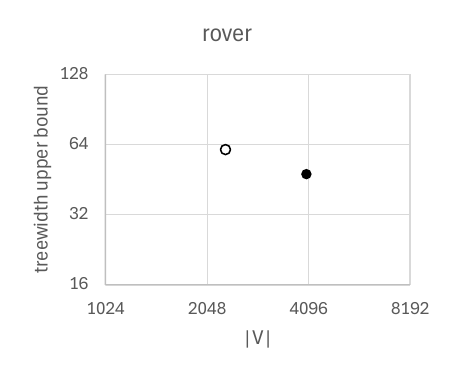}
    \includegraphics[scale=0.8]{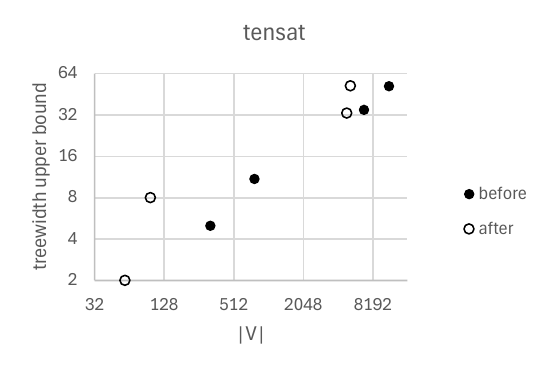}
    \caption{Treewidth and $|V|$ before and after applying all simplification rules. Note that there are some artifacts due to heuristics used in tree decomposition, e.g. in ``rover''. A few examples continued to time out after 15 seconds and are omitted from these data.}
    \label{fig:simp}
\end{figure}

\begin{table}
\centering
\begin{tabular}{c||c|c|c|c}
\textbf{Source} & \textbf{avg.    $\Delta$ treewidth} & \textbf{avg. $\Delta |V|$} & \textbf{avg. $\Delta |E|$} & \textbf{\% timeout} \\ \hline \hline
babble          & -5\%                                & -64\%                      & -57\%                      & 1\%                 \\ \hline
egg             & -40\%                               & -72\%                      & -80\%                      & 7\%                 \\ \hline
eggcc-bril      & -65\%                               & -97\%                      & -96\%                      & 17\%                \\ \hline
flexc           & 1\%                                 & -74\%                      & -81\%                      & 7\%                 \\ \hline
fuzz            & -46\%                               & -60\%                      & -73\%                      & 0\%                 \\ \hline
rover           & 27\%                                & -42\%                      & -72\%                      & 89\%                \\ \hline
tensat          & -23\%                               & -63\%                      & -64\%                      & 60\%               
\end{tabular}
    \caption{Results of simplification, using heuristic approximate tree decomposition, so some instances may falsely confuse an actual rise in treewidth for an instance that just happens to be more difficult for the heuristic. A few examples continued to time out after 15 seconds and are omitted from these data.}
    \label{tbl:simp}
\end{table}

Lastly, we note that these simplification methods bring many more e-graphs into the range in which treewidth-based methods are faster than existing methods. Though we did not develop an efficient implementation ourselves, \cite{glp24} claims that their implementation of treewidth-based extraction is faster than ILP solvers for most e-graphs with treewidth under 10, at least when not required to output acyclic extractions. Few e-graphs in our test set had treewidth under 10 before simplification, but especially for the egraphs in the collections ``eggcc-bril'' and ``fuzz'', a significant fraction of them have treewidth under 10 after simplification.


\section{Future directions}

Many open questions remain regarding extraction and treewidth. The largest open question in our mind relates to more general cost functions than the simple additive ones we have considered here. This is the one area where ILP solvers can never work---by nature of being integer \emph{linear} programs, they are not suitable for other cost functions.

In general, every e-node can be associated with a local cost function, which depends on the costs of its children. For example, an e-node representing the operation ``do my child again'' could have cost $c(x) = x$, in other words it would copy the cost of its child. Such cost functions present a challenge to this present algorithm because we do not enforce the children to be discovered before the parent. In other words, we have to decide whether or not ``do my child again'' is cheap, without knowing the cost of the child, in order to decide whether or not to keep that partial evaluation.

One potential solution would be to represent costs abstractly, and keep all minimal cost solutions, not just one minimum cost one. In this example, we would define a symbolic variable $x$ denoting the unknown cost, and set the cost of ``do my child again'' to $x$. If an evaluation has cost $x$ and another evaluation with the same summary has, for example, cost 7, we would keep both, and then evaluations with cost $2x$, $x+5$, or 8 could be thrown away. When the vertex associated with $x$ is inserted, these expressions would update based on the cost of $x$ that we now know. However, this would massively increase the running time of the algorithm, and would generally not be polynomial with fixed treewidth. It is an interesting open direction to hand such general cost functions efficiently.

We note that \cite{glp24} analyzed the main algorithm (that we share) closely, and identified some criteria slightly more general than additive cost functions that actually work automatically. However, those criteria still exclude nodes like ``do my child again'', so there remains work to be done.

A second open direction is to investigate more closely if particular methods of e-graph creation lead to smaller treewidth, and then design e-graph creation methods (or \emph{saturation}, as it is often called) that minimize treewidth from the start. Since treewidth is often calculated with heuristic algorithms, it would be interesting to determine if certain saturation methods could even be designed with a particular treewidth heuristic in mind to further improve efficiency.


\appendix
\section{Proof of \cref{prop:equiv}} \label{sec:appendix}
\begin{proposition}
    Let $\mathcal{G} = (N, \mathcal{C}, \mathcal{E}, \mathcal{C}_{out}, c)$ be an e-graph and define $G = (V, E, V_{out}, g, c)$ as follows:
    \begin{enumerate}
        \item Let $V = \{x_u : u \in N \} \cup \{\land_u : u \in N\} \cup \{\lor_C : C \in \mathcal{C}\}$.
        \item Let 
        \begin{equation*}
            E = \{(\land_u, \lor_C) : u \in C \in \mathcal{C}\} \cup \{(\lor_C, \land_u) : (u, C) \in \mathcal{E}\} \cup \{(x_u, \land_u) : u \in N\} \tag{$\star$}
        \end{equation*}
        
        \item Let $V_{out} = \{\lor_C : C \in \mathcal{C}_{out}\}$.
        \item Let $g$ map each $\land_u$ to AND and each $\lor_C$ to OR.
        \item Let $c(x_u) = c(u)$.  
    \end{enumerate}

Then the following map is a bijection between minimal satisfiable extractions of $\mathcal{G}$ and minimal satisfiable evaluations of $G$. We map an extraction $\varphi$ to the function $\alpha$ defined for all $u \in C \in \mathcal{C}$:
\begin{equation*}
    \alpha(x_u) = \alpha(\land_u) = \mathbf{1}(\varphi(C) = u) \qquad \qquad \alpha(\lor_C) = \mathbf{1}(C \in \dom(\varphi)) \tag{$\dagger$}
\end{equation*}
The bijection also preserves acyclicity and cost.
\end{proposition}

\begin{proof}
First, we need to show that $\alpha$ is a minimal satisfying evaluation.
\begin{itemize}
    \item To show that $\alpha$ is valid, we check the gate functions using ($\star$) and ($\dagger$).
    \begin{enumerate}
        \item If $\alpha(\land_u) = 0$, then $x_u$ is an input and $\alpha(x_u) = 0$, so we are good.
        \item If $\alpha(\land_u) = 1$, then $\alpha(x_u) = 1$, and the other inputs are $\lor_C$ for all $(u, C) \in \mathcal{E}$. By definition of extraction, $C \in \dom(\varphi)$, so $\alpha(\lor_C) = 1$ as required.
        \item If $\alpha(\lor_C) = 0$, then $C \not \in \dom(\varphi)$, so for all $u \in C$, we have $\alpha(\land_u) = 0$ as required.
        \item If $\alpha(\lor_C) = 1$, then $C \in \dom(\varphi)$ and $\varphi(C) = u$ for some $u \in C$. Hence $\alpha(\land_u) = 1$ as required.
    \end{enumerate}
    
    \item To show that $\alpha$ is minimally satisfying, it is certainly satisfying from the definitions of $\alpha$ and $V_{out}$, so it remains to show minimality. Suppose for contradiction that there exists a proper subset $A \subsetneq \{ i \in V : \alpha(i) = 1\}$ such that $\alpha|_{A}$ is a satisfying evaluation. 
    
    Let $\mathcal{A} = \{C \in \mathcal{C} : \alpha|_A(\lor_C) = 1\}$. We will show that $\mathcal{A} \subsetneq \dom(\varphi)$ yet $\varphi|_{\mathcal{A}}$ is satisfying, contradicting the fact that $\varphi$ is minimal satisfying. 
    
    First, note that $\mathcal{A}$ is a subset of $\dom(\varphi)$ because $\alpha|_A(\lor_C) = 1$ implies $C \in \dom(\varphi)$ by ($\dagger$). Next, we show that it is a proper subset. Because $A$ is a proper subset of $\{i \in V : \alpha(i) = 1\}$, there is a true vertex (by $\alpha$) outside of $A$. 
    \begin{enumerate}
        \item If it is $\lor_C$, then $\alpha|_A(\lor_C) = 1$ and $C$ is a satisfied class outside of $\mathcal{A}$.
        \item If it is $\land_u$, by our construction ($\dagger$), no siblings of $\land_u$ are true. So if $\land_u \not \in A$ and $C$ denotes the class containing $u$, by validity of $\alpha|_A$, we have $\alpha|_A(\lor_C) = 0$, and $C$ is a satisfied class outside of $\mathcal{A}$. 
        \item If it is $x_u$, by validity of $\alpha|_A$, the vertex $\land_u$ must be false in $\alpha|_A$, and $\land_u$ must be true in $\alpha$ by ($\dagger$). Therefore, we are in case (2) and can repeat the argument to find the desired $C$.
    \end{enumerate}

Now to show that $\varphi|_\mathcal{A}$ is a satisfying extraction. First, to show that it is a valid extraction, if $\varphi(C)$ depends on $C_1, \dots, C_k$, by validity of $\alpha|_A$ it must be that $\alpha|_A(\lor_{C_i}) = 1$ for each $i$, so $C_i \in \mathcal{A}$ as required for an extraction. The extraction is satisfying because $\alpha|_A$ being satisfying implies $\alpha|_A(\lor_C) = 1$ for all $C \in \mathcal{C}_{out}$, and hence $C \in \mathcal{A}$ as required. 
\end{itemize}

To prove that the map is a bijection, we define the inverse. Given a minimal satisfying total evaluation $\alpha$, let $\varphi(C) = u$ if and only if $\alpha(\land_u) = 1$ for some $u \in C$. The inverse map is a well-defined choice function because for every $C$, there exists at most one $u \in C$ such that $\alpha(\land_u) = 1$. This follows from minimality of $\alpha$: if $\alpha(\land_u) = \alpha(\land_v) = 1$ for $u, v \in C$, then taking $A = V \setminus \{\land_u, x_u\}$ would allow $\alpha|_{A}$ to be a satisfying evaluation, noting that the only output of $\land_u$ is $\lor_C$. Now we need to show that $\varphi$ is a minimal satisfying extraction.

\begin{itemize}
    \item To show that $\varphi$ is an extraction, simply note that whenever $\varphi(C)$ is defined, by the above construction $\alpha(\land_{\varphi(C)}) = 1$, so by validity of $\alpha$ and ($\star$), all dependencies $C_i$ must have $\alpha(\lor_{C_i}) = 1$. Again by validity, this means that for each $C_i$, at least one $u_i \in C_i$ must have $\alpha(\land_{u_i}) = 1$, so $C_i \in \dom(\varphi)$ and we are done.
    \item To show that $\varphi$ is satisfying, simply note that because $\alpha$ is satisfying, $\alpha(\lor_C) = 1$ for all $C \in \mathcal{C}_{out}$. Each $\lor_C$ can only be 1 if at least one of its children $\land_u$ is evaluated to 1 where $u \in C$, so $C \in \dom(\varphi)$ and we are done.

    To show minimality, suppose for contradiction that there exists a proper subset $\mathcal{A} \subsetneq \dom(\varphi)$ such that $\varphi|_\mathcal{A}$ is a satisfying extraction. Consider $A = \{x_u, \land_u, \lor_C : \varphi|_{\mathcal{A}}(C) = u\}$. We will show that $\alpha|_A$ is a satisfying evaluation with $A \subsetneq V$. 

    To show that $\alpha|_A$ is a valid evaluation, we need to check the gate functions.
    \begin{enumerate}
        \item If $\alpha|_A(\land_u) = 0$, then either $\alpha(\land_u) = 0$, in which case this is valid by validity of $\alpha$, or the class of $u$ did not belong to $\mathcal{A}$, in which case we also have $\alpha|_A(x_u) = 0$, which suffices for validity.
        \item If $\alpha|_A(\land_u) = 1$, then where $C$ is the class of $u$, we have $C \in \mathcal{A}$. Because $\varphi|_{\mathcal{A}}$ an extraction, we conclude that the children $C_1, \dots, C_k$ of $u$ also belong to $\mathcal{A}$, so combined with the fact that $\alpha$ is valid, we conclude that $\alpha|_A(\lor_{C_i}) = 1$ as desired. Again because $\alpha$ is valid and $\land_u \in A$ if and only if $x_u \in A$, we also have $\alpha|_A(x_u) = 1$ to conclude.
        \item If $\alpha|_A(\lor_C) = 0$, then either $\alpha(\lor_C) = 0$, in which this is valid by validity of $\alpha$, or $C \not \in \mathcal{A}$, in which case there is no $u \in C$ for which $\land_u \in A$. Hence $\alpha|_A(\land_u) = 0$ for all $u \in C$, which suffices for validity.
        \item If $\alpha|_A(\lor_C) = 1$, then we have $C \in \mathcal{A}$. Then where $u = \varphi|_\mathcal{A}(C) = \varphi(C)$, we have $\land_u \in A$, and hence $\alpha|_A(\land_u) = \alpha(\land_u) = 1$, which suffices for validity.
    \end{enumerate}
    
    As mentioned above, $\alpha|_A$ is satisfying because $\alpha$ is satisfying and $\lor_C \in A$ for all $C \in \mathcal{C}_{out}$, because $\varphi|_\mathcal{A}$ is satisfying. Then $A \subsetneq V$ because if $C \in \dom(\varphi) \setminus \mathcal{A}$, then $\alpha(\lor_C) = 1$ whereas $\alpha|_{A} (\lor_C) = 0$. 
\end{itemize}

Lastly, to show that inverse map is truly an inverse, we need to show that transforming $\varphi$ to $\alpha$ to $\varphi$ is the identity, which is obvious, and that transforming $\alpha$ to $\varphi$ to $\alpha$ is identity, for which it suffices to note that $\alpha$ is entirely determined by its values on $\land_u$: vertices $\lor_C$ only have vertices of type $\land_u$ as inputs, so they are determined, and we must have $\alpha(x_u) = \alpha(\land_u)$, because $\alpha(x_u) = 0$ with $\alpha(\land_u) = 1$ is not valid and $\alpha(x_u) = 1$ with $\alpha(\land_u) = 0$ is not minimal (take $A = V \setminus \{x_u\}$).

To show that the bijection preserves acyclicity, we need to show both directions:
\begin{itemize}
    \item Suppose $\varphi$ is acyclic. To show that $G[\alpha]$ is acyclic, it suffices to show that every directed cycle in $G$ has at least one vertex on which $\alpha$ is 0. The only possible directed cycles of $G$ occur as alternations of edges of type $(\land_u, \lor_C)$ and $(\lor_C, \land_u)$. So let $\land_{u_1}, \lor_{C_1}, \dots, \land_{u_k}, \lor_{C_k}, \land_{u_{k+1}} = \land_{u_1}$ be a cycle in $G$, where $u_i \in C_i$ and $(u_{i+1}, C_i) \in \mathcal{E}$ for all $i$, and it would suffice to find $u_i$ such that $\alpha(\land_{u_i}) = 0$. Because $\varphi$ is acyclic, there must be some $u_i$ for which $\varphi(C_i) \neq u_i$ (otherwise $C_1, \dots, C_k, C_1$ is a cycle). Therefore by ($\dagger$), $\alpha(\land_{u_{i}}) = 0$ as desired.  
    \item Suppose $\alpha$ is acyclic and suppose for contradiction that $C_1, \dots, C_k = C_1$ forms a cycle in $\varphi$. Then for each $i$, denoting $u_i = \varphi(C_i)$, we have that $\alpha(\land_{u_i}) = 1$ and hence $\alpha(\lor_{C_i}) = 1$. But then $\land_{u_1}, \lor_{C_1}, \dots, \land_{u_{k-1}}, \lor_{C_{k-1}}, \land_{u_1} = \land_{u_k}$ is a cycle in $G[\alpha]$, a contradiction. 
\end{itemize}

Lastly, for the cost, we simply note that by the bijection, $x_u = 1$ if and only if $\varphi(C) = u$ where $C$ is the class of $u$, so this is clear.
\end{proof}

\bibliographystyle{alpha} 
\bibliography{bib}

\end{document}